\documentclass{amcjou}

\usepackage{amsmath,amssymb,amsthm,amsfonts}
\usepackage{subfigure}
\usepackage[usenames,dvipsnames]{pstricks}
\usepackage{epsfig}
\usepackage{pst-grad} 
\usepackage{pst-plot} 
\usepackage{thmbox}
\newcommand{\For}{{\texttt{{FOR}}}}
\newcommand{\Do}{\texttt{{DO}}}
\newcommand{\If}{\texttt{{IF}}}
\newcommand{\Else}{\texttt{{ELSE}}}
\newcommand{\Then}{\texttt{{THEN}}}

\newcommand{\notalpha}{\overline\alpha}
\newcommand{\R}{\mathfrak d}
\newcommand{\la}{\langle}
\newcommand{\ra}{\rangle}

\renewcommand{\square}{\Box}
\newcommand{\mh}[1]{{#1}}
\newcommand{\tk}[1]{{#1}}
\newcommand{\w}[1]{{#1}}

\newtheorem[M]{algo}{Algorithm}[section]

\begin{document}
\sloppy

\begin{frontmatter}   

\titledata{Fast Recognition of Partial Star Products and Quasi Cartesian Products}    
{ We thank Lydia Ostermeier for her insightful comments on graph bundles,
	as well as for \w{the suggestion of} the term "quasi product".
	This work was supported in part by the Deutsche Forschungsgemeinschaft
 (DFG) Project STA850/11-1 within the EUROCORES Program EuroGIGA (project
 GReGAS) of the European Science Foundation.
 This paper is based on part of the dissertation of  the third author.}                 

\authordata{Marc Hellmuth}            
{Center for Bioinformatics,
Saarland University,
D - 66041 Saarbr\"{u}cken,
Germany \\
}    
{marc.hellmuth@bioinf.uni-sb.de}                     
{}

\authordata{Wilfried Imrich}            
{Chair of Applied Mathematics,
Montanuniversit{\"a}t,  A-8700
Leoben,
Austria}
{imrich@unileoben.ac.at}
{}                                       

\authordata{Tomas Kupka}            
{Department of Applied Mathematics,
VSB-Technical University of Ostrava,
Ostrava, 70833, Czech Republic\\
}
{tomas.kupka@teradata.com}
{}


\begin{abstract}
This paper
is concerned with the fast computation of a relation \w{$\R$} on the edge
set of connected graphs that plays a decisive role in the recognition of
approximate Cartesian products, the weak reconstruction of Cartesian
products, and the recognition of Cartesian graph bundles with a triangle
free basis.

A special case \w{of $\R$} is the relation $\delta^\ast$, whose convex closure yields
the \w{product relation $\sigma$ that induces the} prime factor decomposition of
connected graphs with respect to the Cartesian product.
For the construction of \w{$\R$} so-called Partial
Star Products are of particular interest.  Several special data structures are used
 that allow to compute Partial Star Products in constant time.
\w{These computations are tuned to the recognition of approximate graph products,
but also lead to} a linear time
algorithm for the computation of $\delta^\ast$ for graphs with maximum
bounded degree.

\w{Furthermore},  we define \emph{quasi Cartesian products} as graphs with
non-trivial $\delta^\ast$. We provide several examples, and show that
quasi Cartesian products can be recognized in linear time for graphs
with bounded maximum degree. \w{Finally}, we  note that
quasi products can be recognized in sublinear time with a parallel\w{ized} algorithm.

\end{abstract}

\end{frontmatter}   


\section{Introduction}
Cartesian products of graphs derive their popularity from their simplicity,
and their importance from the fact that many classes of graphs, such as
hypercubes, Hamming graphs, median graphs, benzenoid graphs, or Cartesian
graph bundles, are either Cartesian products or closely related to them
\cite{haimkl-2011}. As even slight disturbances of a product, such as
the addition or deletion of an edge, can destroy the product structure
completely \cite{Fei-1986}, the question arises whether it is possible to
restore the original product structure after such a disturbance. In other
words, given a graph, the question is, how close it is to a Cartesian
product, and whether one can find this product algorithmically.
Unfortunately, in general this problem can only be solved by heuristic
algorithms, as discussed in detail in \cite{HeImKu-2013}.
\w{That} paper
also presents several heuristic \w{algorithms} for the solution of this problem.

One of the main steps towards such algorithms is the computation of an
equivalence relation $\R_{|S_v}(W)^*$ on the edge-set of a graph. The
complexity \w{of} the computation of $\R_{|S_v}(W)^*$ in \cite{HeImKu-2013} is
$O(n\Delta^4)$, where $n$ is the number of vertices, and $\Delta$ the
maximum degree of $G$. Here we improve the recognition complexity of
$\R_{|S_v}(W)^*$ to $O(m\Delta)$, where $m$ is the number of edges of
$G$, \w{and thereby improve} the complexity of the just mentioned heuristic
algorithms.

A special case is the computation of the relation $\delta^\ast =
\R_{|S_v}(V(G))^*$. This relation defines the so-called quasi Cartesian
product, see Section \ref{sec:twist}. Hence, quasi products can be
recognized in $O(m \Delta)$ time. As the algorithm can easily be
parallelized, it leads to sublinear recognition of quasi Cartesian
products.

When the given graph $G$ is a Cartesian product from which
just one vertex was deleted, things are easier. In that case, the product
is uniquely defined and can be reconstructed in polynomial time from $G$,
see \cite{do-1975} and \cite{haze-1999}. In other words, if $G$ is given,
and if one knows that there is a Cartesian product graph $H$ such that $G =
H\smallsetminus x$, then $H$ is uniquely defined. Hagauer and \v{Z}erovnik
showed that the complexity of finding $H$ is $O(mn(\Delta^2+m))$. The
methods of the present paper will lead to a new algorithm of complexity
$O(m\Delta^2+\Delta^4)$ for the solution of this problem. This is
 \w{part of the dissertation} \cite{ku-2013} \w{of the third author}, and will be the topic of a subsequent publication.

Another class of graphs that is closely related to Cartesian products are
Cartesian graph bundles, see Section \ref{sec:twist}. In \cite{impize-1997}
it was \w{proved} that Cartesian graph bundles over a triangle-free base can be
effectively recognized, and in \cite{PiZmZe-2001} it was shown that this
can be done in $O(mn^2)$ time. With the methods of this paper, we suppose
that one can improve it to $O(m\Delta)$ time.
This too will be published separately.


\section{Preliminaries}
We consider finite, connected undirected graphs $G=(V,E)$ without loops and
multiple edges. The {\it Cartesian product} $G_1\square G_2$ of graphs
$G_1=(V_1,E_1)$ and $G_2=(V_2,E_2)$ is a graph with vertex set
$V_1\times V_2$, where the vertices $(u_1,v_1)$ and $(u_2,v_2)$ are
adjacent if $u_1u_2\in E_1$ and $v_1=v_2$, or if
$v_1v_2\in E_2$ and $u_1=u_2$. The Cartesian
product is associative, commutative, and has the one vertex graph $K_1$ as
a unit \cite{haimkl-2011}.
By associativity we can write $G_1\square \cdots \square G_k$ for a product
$G$ of graphs $G_1,\ldots,$ $G_k$ and can label the vertices of $G$ by the
set of all $k$-tuples $(v_1, v_2,\dots,v_k)$, where $v_i\in G_i$ for $1\le
i\le k$. One says two edges have the same \emph{Cartesian color} if their
endpoints differ in the same coordinate.

A  graph $G$ is {\it prime} if the identity $G=G_1\square G_2$
implies that $G_1$ or $G_2$ is the one-vertex graph $K_1$. A representation
of a graph $G$ as a product $G_1\square G_2\square\dots\square G_k$ of
\mh{non-trivial} prime graphs is called a {\it prime factorization} of $G$.
It is well known
that every connected graph $G$ has a prime factor decomposition with
respect to the Cartesian product, and that this factorization is unique up
to isomorphisms and the order of the factors, see Sabidussi \cite{sa-1960}.
Furthermore, the prime factor decomposition can be computed in linear time,
see \cite{impe-2007}.

Following the notation in \cite{HeImKu-2013}, an induced cycle
on four vertices is called \emph{chordless square}. Let the edges $e=vu$
and $f=vw$ span a chordless square $vuxw$. Then $f$ is the
\emph{opposite} edge of the edge $xu$. The vertex $x$ is called \emph{top vertex}
(w.r.t. the square spanned by $e$ and $f$). A top vertex $x$ is
\emph{unique} if $|N(x) \cap N(v)|= 2$, where $N(u)$ denotes the (open)
$1$-neighborhood of vertex $u$. In other words, a top vertex $x$ is not
unique if there are further squares with top vertex $x$ spanned by the
edges $e$ or $f$ together with a third distinct edge $g$.
\tk{Note that the existence of a unique top vertex $x$ does not imply
	that $e$ and $f$
span a unique square, as there might be another $vuyw$ with a possible
unique top vertex $y$. Thus, $e$ and $f$ span
unique square $vuxw$ only if $|N(u) \cap N(w)| =2$ holds.}
The \emph{degree}
$\deg(u):=|N(u)|$ of a vertex $u$ is the number of edges that
contain $u$. The maximum degree of a graph is denoted by $\Delta$ and
a path on $n$ vertices by $P_n$.

We now \w{recall the Breadth-First Search (BFS)} ordering of the vertices $v_0,\ldots,v_{n-1}$ of a
graph: select an arbitrary, but
fixed vertex $v_0\in V(G)$, called the {\it root}, and create a sorted list
of vertices. \w{Begin} with $v_0$; append all neighbors
$v_1,\ldots,v_{\deg(v_0)}$ of $v_0$ to the list; then append all neighbors of $v_1$
that are not already in \w{the} list; \w{and} continue recursively with
$v_2,v_3,\ldots$ until all vertices of $G$ are processed.

\subsection{The Relations $\boldsymbol{\delta}$, $\boldsymbol{\sigma}$ and the Square Property.}

There are two basic relations defined on an edge set of a given graph
that play an important role in the field of Cartesian product recognition.

\begin{defn}
Two edges $e,f\in E(G)$ are in the \emph{relation $\delta_G$},
if one of the following conditions in $G$ is satisfied:
\begin{itemize}
\item[(i)]  $e$ and $f$ are adjacent and it is not the case that there is a
           unique square spanned by $e$ and $f$, and that this square is
           chordless.
\item[(ii)]  $e$ and $f$ are opposite edges of a chordless square.
\item[(iii)] $e=f$.
\end{itemize}
\label{def:delta}
\end{defn}

Clearly, this relation is  reflexive and symmetric but not
necessarily transitive. We denote its {\it transitive closures}, that is,
the smallest transitive relation containing $\delta_G$, by
$\delta_G^*$. \\

If adjacent edges $e$ and $f$ are not in relation $\delta$, that is,
\w{if} Condition (i) of  Definition \ref{def:delta} is not fulfilled, then
they span a unique \w{square, and this square is} chordless. We call such a square
just \emph{unique chordless square (spanned by $e$ and $f$)}.\\

Two edges $e$ and $f$ are in the \emph{product relation} $\sigma_G$ if they have the
same Cartesian colors with respect to the prime factorization of $G$.
The product relation $\sigma_G$ is a uniquely defined
equivalence relation on $E(G)$ that contains all information about the
prime factorization\footnote{For the properties of $\sigma$ that we will
cite or use, we refer the reader to \cite{haimkl-2011} or \cite{imkl-2000}.}.
Furthermore, $\delta_G$ and $\delta_G^\ast$ are contained in $\sigma_G$.

\w{Notice that we may also use the notation $\delta(G)$ for $\delta_G$, respectively $\sigma(G)$ for $\sigma_G$.}
If there is no risk of confusion we \w{may even} write $\delta$, resp., $\sigma$.
instead of $\delta_G$, resp., $\sigma_G$.

We say an equivalence relation $\rho$ defined on the edge set of a graph
$G$ has the {\it square property} if the following three conditions hold:
\begin{itemize}
	\item[(a)] For any two edges $e = uv$ and $f = uw$ that belong to
	           different equivalence classes of $\rho$ there exists a unique
	           vertex $x \neq u$ of $G$ that is adjacent to $v$ and
	           $w$. 
   \item[(b)] The square $uvxw$ is chordless.
	\item[(c)] The opposite edges of any chordless square belong to the same equivalence class of $\rho$.
\end{itemize}

From the definition of $\delta$ it easily follows that $\delta$ is a
refinement of any such $\rho$. It also implies that $\delta^\ast$, and thus
also $\sigma$, have the square property. This property is of fundamental
importance, both for the Cartesian and the quasi Cartesian product.
We note in passing that $\sigma$ is the convex hull of
$\delta^\ast$, see \cite{imze-1994}.


\subsection{The Partial Star Product}
\label{sec:psp}
This section is concerned with the {\it partial star product}, which plays
a decisive role in the local approach. As it was introduced in
\cite{HeImKu-2013}, we will only \w{define it here,} list some of its most basic properties\w{,}
and refer to \cite{HeImKu-2013} for details.

Let $G=(V,E)$ be a given graph
and $E_v$ the set of all edges incident to some vertex $v \in V$. We
define the local relation $\R_v$ as follows: $$\R_v = ((E_v\times E) \cup
(E\times E_v)) \cap \delta_G \subseteq \delta(\la N_2^G[v]\ra).$$ In other
words, $\R_v$ is the subset of $\delta_G$ that contains all pairs $(e,f)\in
\delta_G$, where at least one of the edges $e$ and $f$ is incident to $v$.
\w{Clearly} $\R^*_v$, which is not necessarily a subset of
$\delta$, is contained in $\delta^*$, \w{see} \cite{HeImKu-2013}.

Let $S_v$ be a subgraph of $G$ that contains all edges incident to $v$ and
all squares spanned by edges $e, e'\in E_v$ where $e$ and $e'$ are not in
relation $\R_v^*$. Then $S_v$ is called \emph{partial star product} (PSP
for short). To be more precise:

\begin{defn}[Partial Star Product (PSP)]
Let $F_v\subseteq E\setminus E_v$
be the set of edges which are opposite edges of (chordless)
squares spanned by $e,e'\in E_v$ that are in different
$\R^*_v$ classes, that is, $(e,e') \not\in \R^*_v$.

Then the \emph{partial star product} is the subgraph
$S_v \subseteq G$ with edge set $E'= E_v\cup F_v$ and vertex set $\cup_{e\in E'}
e$. We call $v$ the \emph{center} of $S_v$, $E_v$ the set of \emph{primal
edges}, $F_v$ the set of \emph{non-primal edges}, and the vertices adjacent
to $v$ \emph{primal vertices} of $S_v$.
	\label{def:starproduct}
\end{defn}

As shown in \cite{HeImKu-2013}, a partial star product $S_v$ is always an
isometric subgraph or even isomorphic to a Cartesian product graph $H$, where
the factors of $H$ are so-called stars $K_{1,n}$.
These stars can directly be determined by the respective $\R_v^*$ classes,
see \cite{HeImKu-2013}.

Now we define a {\it local coloring} of $S_v$ as the restriction of
the relation $\R^*_v$ to $S_v$:
$$\R_{|S_v}:=\R^*_{v|S_v}= \{(e,f) \in \R^*_v \mid e,f \in E(S_v)\}.$$
In other words, $\R_{|S_v}$ is the subset
of $\R^*_v$ that contains all pairs of edges $(e,f)\in \R^*_v$ where both
$e$ and $f$ are in $S_v$ and edges obtain the same local color whenever
they are in the same equivalence class of $\R_{|S_v}$.
As an example consider the PSP $S_v$ in Figure \ref{fig:agp}.
The relation $\R_{|S_v}$ has three equivalence classes
(highlighted by thick, dashed and double-lined edges).
Note, $\delta^*$ just contains one equivalence class. Hence,
$\R_{|S_v} \neq \w{\delta(S_v)^*}$.

For a given subset $W\subseteq V$ we set
$$\R_{|S_v}(W) = \cup_{v\in W} \R_{|S_v}\,.$$ The transitive closure of
$\R_{|S_v}(W)$ is then called the {\it global coloring} with respect to
$W$. As shown in \cite{HeImKu-2013}, we have the following theorem.

\begin{thm}
	Let $G=(V,E)$ be a given graph and
	$\R_{|S_v}(V) = \cup_{v_\in V} \R_{|S_v}$.
	Then $$\R_{|S_v}(V)^* = \delta(G)^*.$$
\label{thm:union_equals_delta}
\end{thm}

For later reference and for the design of the recognition algorithm
we \w{list} the following three lemmas about \w{relevant} properties of the PSP.

\begin{lem}[\cite{HeImKu-2013}]
\label{lem:PSP1}
Let G=(V,E) be a given graph and $S_v$ be a PSP  \w{of} an arbitrary vertex $v\in V$.
If  $e, f \in E_v$ are primal edges that are not in relation $\R_v^*$, then
$e$ and $f$ span a unique chordless square with a unique top vertex in $G$.

Conversely, suppose that $x$ is a non-primal vertex of $S_v$. Then
there is a unique chordless square in $S_v$ that contains $x$,
and that is spanned by edges $e, f \in E_v$ with $(e,f)\not\in \R^*_v$.
\end{lem}

\begin{lem}[\cite{HeImKu-2013}]
Let G=(V,E) be a given graph and $f \in F_v$ be a non-primal edge of a
PSP $S_v$ \w{of} an arbitrary vertex $v\in V$. Then $f$ is opposite to exactly one
primal edge $e\in E_v$ in $S_v$, and $(e,f)\in \R_{|S_v}$.
\label{lem:PSP2}
\end{lem}

\begin{lem}[\cite{HeImKu-2013}]
Let G=(V,E) be a given graph and $W\subseteq V$ such that $\la W \ra$ is connected.
Then each vertex $x\in W$ meets every equivalence class
of $\R_{|S_v}(W)^*$ in $\cup_{v\in W} S_v$.
\label{lem:vMeetsEveryClass}
\end{lem}


\section{Quasi Cartesian Products}
\label{sec:twist}

Given a Cartesian product $G = A \square B$ of two connected, prime graphs
$A$ and $B$\w{,} one can recover the factors $A$ and $B$ as follows: the product
relation $\sigma$ has two equivalence classes, say $E_1$ and $E_2$, and the
connected components of the graph $(V(G), E_1)$ are all isomorphic copies
of the factor $A$, or of the factor $B$, see Figure \ref{fig:cp}. This property
naturally extends to products of more than two prime factors.

We already observed that $\delta$ is finer than any equivalence relation
$\rho$ that satisfies the square property. Hence the equivalence classes of
$\rho$ are unions of $\delta^\ast$-classes. This also holds for $\sigma$.
It is important to \w{keep in mind that} $\sigma$  can \mh{be trivial, that
is, it} consists of a single equivalence
class \w{even when $\delta^\ast$ has} more than one equivalence class.

We call all graphs $G$ with a non-trivial equivalence relation $\rho$ that is
defined on $E(G)$ and satisfies the square property \emph{quasi
(Cartesian) products}. Since $\delta^*\subseteq \rho$ for every such
relation $\rho$, it follows that $\delta^*$ must have at least two
equivalence classes for any quasi product. By Theorem
\ref{thm:union_equals_delta} \w{we have} $\R_{|S_v}(V(G))^* =\delta^*$.
In other words, quasi products can be defined as graphs where the PSP's of
all vertices are non-trivial, that is, none of the PSP's is a star $K_{1,n}$, and in
addition, where the union over all \w{$\R_{|S_v}$ yields} a non-trivial $\delta^*$.

Consider the equivalence classes of the relation $\delta^\ast$ of the
graph $G$ of Figure \ref{fig:twist2}.
It has two equivalence classes, and locally looks
like a Cartesian product, but is actually reminiscent of a M\"obius band.
Notice that the graph $G$ in Figure \ref{fig:twist2} is prime with respect to
Cartesian multiplication, although $\delta^*$ has two equivalence classes:
all components of the first class are paths of length $2$, and there are
two components of the other $\delta^\ast$-class, which do not
have the same size. Locally this graph looks
either like $P_3 \square P_3$ or $P_2 \square P_3$.

\begin{figure}[tbp]
\centering
  \subfigure[The Cartesian product $G=P_3\Box C_4$.]{
    \label{fig:cp}			
		\includegraphics[bb= 174 483 417 660, width=0.4\textwidth]{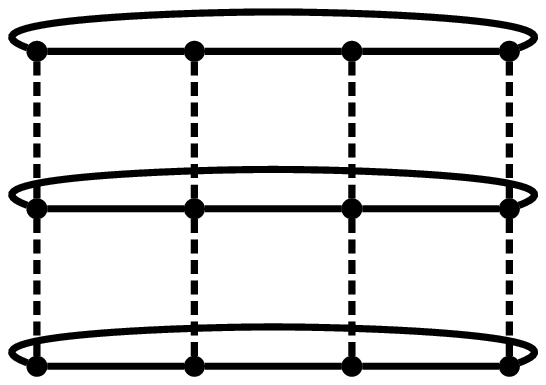}
  }$\qquad\qquad$
  \subfigure[A quasi Cartesian product, which is also a graph bundle.]{
    \label{fig:twist2} 
    \includegraphics[bb= 174 483 417 660, width=0.4\textwidth]{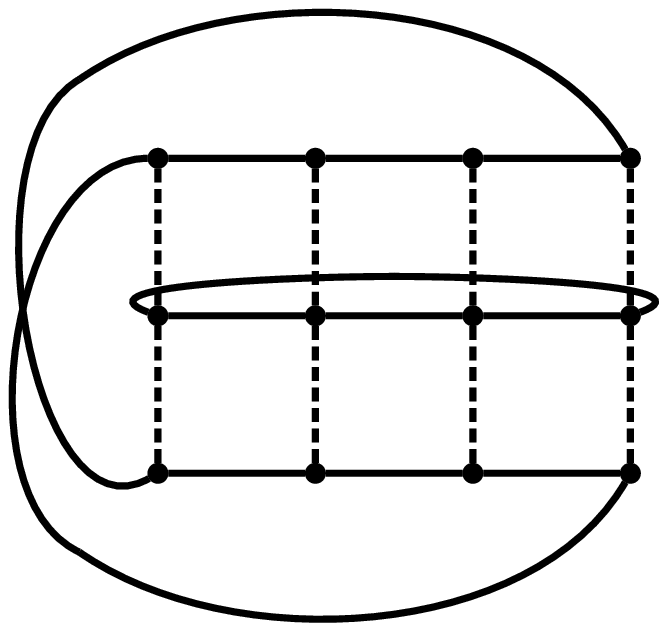}
  }
  \subfigure[A quasi Cartesian product, which is not a graph bundle.]{
    \label{fig:twistNotBundle} 
    \includegraphics[bb=174 455 417 688, width=0.4\textwidth]{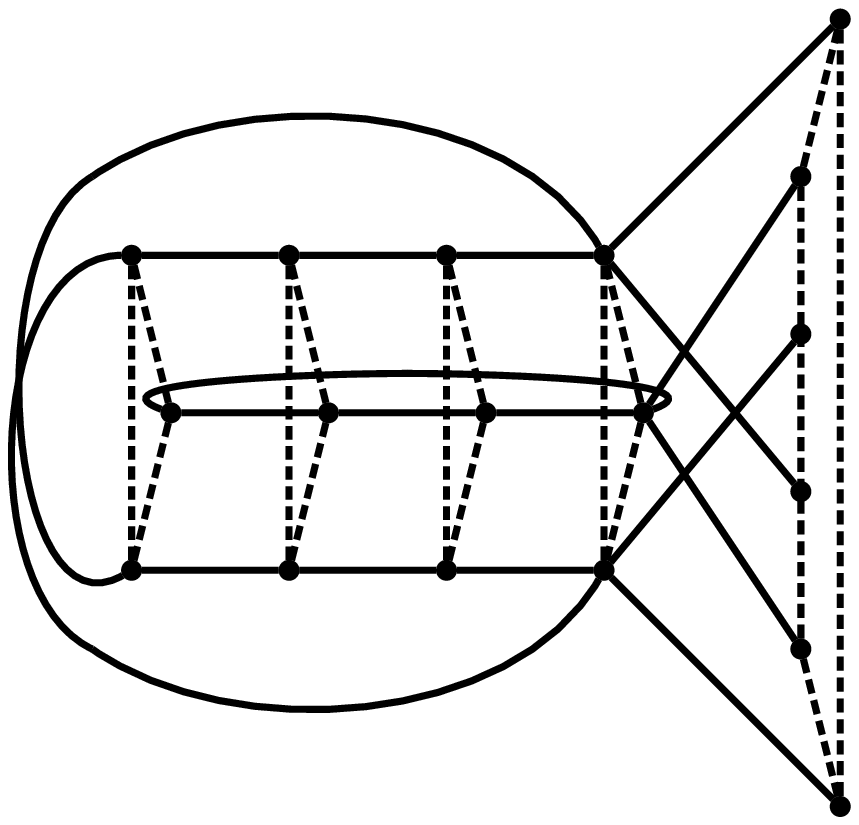}
  }
  $\qquad\qquad$
  \subfigure[The approximate product and PSP $S_v$, which is neither a quasi product nor a graph bundle.]{
    \label{fig:agp}   
		  \includegraphics[bb=  55 443 298 626, scale=0.6]{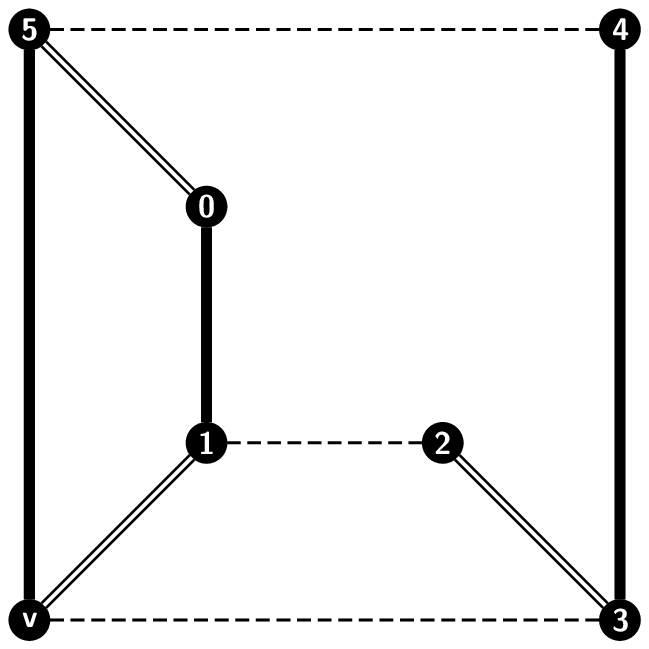}
  }
  \caption{Shown are several quasi Cartesian products, graph bundles and approximate products.}
  \label{fig:Exmpl1}
\end{figure}

In fact, the graph in Figure \ref{fig:twist2} is a so-called Cartesian
graph bundle \cite{impize-1997}, where \emph{Cartesian graph bundles} are defined as follows:
Let $B$ and $F$ be graphs. A graph $G$ is a (Cartesian) graph bundle with fiber $F$ over the base $B$
if there exists a weak homomorphism\footnote{A weak
homomorphism maps edges into edges or single vertices.} $p : G \rightarrow
B$ such that
\begin{itemize}
\item[(i)] for any $u \in V (B)$, the  subgraph (induced by) $p^{-1}(u)$ is isomorphic to $F$, and
\item[(ii)] for any $e \in E(B)$, the subgraph $p^{-1}(e)$ is isomorphic to $K_2 \square F.$
\end{itemize}

The graph of Figure \ref{fig:twistNotBundle} shows that not all quasi
Cartesian products are graph bundles. On the other hand, not every graph
bundle has to be a quasi product. The standard example is
the  complete bipartite graph $K_{2,3}$.
It is a graph bundle with base $K_3$ and fiber $K_2$, but has
only one $\delta^\ast$-class.

Note, in \cite{HeImKu-2013} we  \w{considered} "approximate products"
which  \w{were} first introduced in \cite{hiks-2008, hellmuth-2011}.
As approximate products are all graphs that have a
\w{(small) edit} distance
to a non-trivial product graph, it is clear that every bundle and quasi product
can be considered as \w{an} approximate product, while the converse is not true.
\w{For example,} consider the graph in Figure \ref{fig:agp}.
Here, $\delta^*$ has only one equivalence class.
However, the relation $\R_{|S_v}$ has, in this case,
three equivalence classes (highlighted by thick, dashed and double-lined edges).


Because of the local product-like structure of quasi Cartesian products
we are led to the following conjecture:
\begin{conj}
Quasi Cartesian products can be reconstructed in essentially the same time
from vertex-deleted subgraphs as Cartesian products.
\end{conj}

\section{Recognition Algorithms}

\subsection{Computing the Local and Global Coloring}

For a given graph $G$, let $W\subseteq V(G)$ be an arbitrary subset of the
vertex set of $G$ such that the induced subgraph $\langle W\rangle$ is
connected. Our approach \w{for the computation} is based on the recognition of all PSP's $S_v$ with
$v\in W$, and \w{subsequent} merging of their local colorings. The subroutine computing
local colorings calls the vertices in BFS-order with respect to an
arbitrarily chosen root $v_0\in W$.

Let \w{now} us briefly introduce several additional notions used in the PSP
recognition algorithm. \w{At the start} of every iteration we assign
pairwise different {\it temporary local colors} to the primal edges of
every PSP. These colors are then merged in subroutine processes to compute
 \emph{local colors} associated with every PSP. Analogously, we
use {\it temporary global colors} that are initially assigned to every edge
incident with the root $v_0$.

For any vertex $v$  \w{of} distance two from a PSP center $c$ we store
attributes called {\it first and second primal neighbor}\w{, that is,}
references to adjacent primal vertices from which $v$ was "visited" (in
pseudo-code attributes are accessed by $v.FirstPrimalNeighbor$ and
$v.SecondPrimalNeighbor$). When $v$ is found to have at least two primal
neighbors we add $v$ to $\mathbb{T}_{c}$\w{,} which is a stack of candidates
for non-primal vertices of $S_{c}$. Finally, we use {\it incidence} and
{\it absence lists} to store recognized squares spanned by primal edges.
Whenever we recognize that two primal edges span a square we put them into
the incidence list.  \w{If} we find out that a pair of primal edges
cannot span a unique chordless square with unique top vertex\w{, then} we put it into
the absence list. Note that the
above structures are local and are always associated with a certain PSP
recognition subroutine (Algorithm \ref{alg:PSPrecognition}).
Finally, we will "map" local colors to temporary global colors via
temporary vectors which helps us to merge local with global colors.

Algorithm \ref{alg:PSPrecognition}  computes a local coloring for given PSP's
and merges it with the global coloring $\R_{|S_v}(W)^*$ where $W\subseteq V(G)$ is the set of
treated centers. Algorithm \ref{alg:deltaRecognition} \w{summarizes} the main control structure of the
local approach.

\clearpage
\begin{algo}[PSP recognition]
\small
\label{alg:PSPrecognition}
\begin{tabular}{ll}
\emph{Input:} & Connected graph $G=(V,E)$, PSP center $c\in V$, global coloring $\R_{|S}(W)^*$, \\
							 & where $W\subseteq V$ is the set of treated centers and where the subgraph induced \\ &
								by $W\cup c$ is connected.\\
\emph{Output:} & New temporary global coloring $\R_{|S}(W\cup c)^*$.
\end{tabular}
\renewcommand{\baselinestretch}{1.1}\small
   \begin{enumerate}
     \item\label{alg:PSPrecognition0} Initialization.
     \item\label{alg:PSPrecognition1} \For\ every neighbor $u$ of $c$ \Do:
       \begin{enumerate}
       	\item\label{alg:PSPrecognition1a} \For\ every neighbor $w$ of $u$ (except $c$) \Do:
       	\begin{enumerate}
	         \item\label{alg:PSPrecognition1ai} \If\ $w$ is primal w.r.t. $c$ \Then\ put pair of primal edges $(cu, cw)$ to absence list.
         	\item\label{alg:PSPrecognition1aii} \Else \If\ $w$ was not visited \Then\ set $w.FirstPrimalNeighbor=u$.
				\item\label{alg:PSPrecognition1aiii} \Else\ ($w$ is not primal and was already visited) \Do:
				\begin{enumerate}
					\item\label{alg:PSPrecognition1aiiiA} \If\ only one primal neighbor $v$ $(v\neq u)$ of $w$ was recognized so far, then \Do:
					\begin{itemize}
						\item Set $w.SecondPrimalNeighbor=u$.
						\item  \If\  $(cu,cv)$ is not in incidence list\w{,} then add $w$ to the stack $\mathbb{T}_c$  and
													add the pair $(cu,cv)$ to incidence list.
						\item \Else\ ($cu$ and $cv$ span more squares) add pair $(cu,cv)$ to absence list.
					\end{itemize}
					\item\label{alg:PSPrecognition1aiiiB} \Else:
					\begin{itemize}
						\item Add all pairs formed by primal edges $cv_1,cv_2,cu$ to absence list\w{,} where $v_1,v_2$ are first and second primal neighbors of $w$.
					\end{itemize}
       		\end{enumerate}
			\end{enumerate}
     \end{enumerate}
     \item\label{alg:PSPrecognition2} Assign pairwise different temporary local colors to primal edges.
     \item\label{alg:PSPrecognition3} \For\ any pair $(cu,cv)$ of primal edges $cu$ and $cv$ \Do:
			\begin{enumerate}
				\item\label{alg:PSPrecognition3a} \If\ $(cu,cv)$ is contained in absence list \Then\ merge temporary local colors of $cu$ and $cv$.
				\item\label{alg:PSPrecognition3b} \If\ $(cu,cv)$ is not contained in incidence list \Then\ merge temporary local colors of $cu$ and $cv$.
			\end{enumerate}
			(Resulting merged temporary local colors determine local colors of primal edges in $S_c$. We will  \w{reference} them in the following steps.).
		\item\label{alg:PSPRecognition4} \For\ any primal edge $cu$ \Do:
			\begin{enumerate}
				\item  \If\ $cu$ was already assigned some temporary global color $d_1$ \Then
				\begin{enumerate}
					\item\label{alg:PSPRecognition4i} \If\ local color $b$ of $cu$ was already mapped to some
																									 temporary global color $d_2$\w{,} where $d_2\neq d_1$\w{,} \Then\ merge $d_1$ and $d_2$.
					\item\label{alg:PSPRecognition4ii} \Else\ map local color $b$ to $d_1$.
				\end{enumerate}
			\end{enumerate}
			\item\label{alg:PSPRecognition5} \For\ any vertex $v$ from stack $\mathbb{T}_c$ \Do:
			\begin{enumerate}
				\item Check local colors of primal edges $cw_1$ and $cw_2$ (where $w_1,w_2$ are first and second primal neighbor of $v$, respectively).
				\item \If\ they differ in local colors \Then
				\begin{enumerate}
					\item \If\ there was defined temporary global color $d_1$ for $vw_1$ \Then
					\begin{enumerate}
						\item \If\ local color $b$ of $cw_2$ was already mapped to some temporary global color $d_2$\w{,} where $d_2\neq d_1$ \Then\ merge $d_1$ and $d_2$.
						\item \Else\ map local color $b$ to $d_1$.
					\end{enumerate}
					\item \If\ there was already defined temporary global color $d_1$ for $vw_2$ \Then:
					\begin{enumerate}
						\item \If\ local color $b$ of $cw_1$ was already mapped to some temporary global color $d_2$\w{,} where $d_2\neq d_1$
									\Then\ merge $d_1$ and $d_2$.
						\item \Else\ map local color $b$ to $d_1$.
					\end{enumerate}
				\end{enumerate}
			\end{enumerate}
		\item\label{alg:PSPRecognition6} Take every edge $e$ of the PSP $S_c$
		                                 that was not colored by any temporary
		                                 global color up to now and assign it
		                                 $d$, where $d$ is the temporary global
		                                 color to which the local color of $e$ or the local color of its opposite primal edge $e'$
                                     was mapped.
\\ (If there is a local color $b$ that was not mapped to any temporary global
color\w{, then} we create \w{a} new temporary global color and assign it to all edges of color $b$).
   \end{enumerate}
\renewcommand{\baselinestretch}{1.}\normalsize
\end{algo}
\normalsize

\bigskip
\begin{algo}[Computation of $\R_{|S_v}(W)^*$]
\small
\label{alg:deltaRecognition}
\begin{tabular}{ll}
\emph{Input:} & A connected graph $G$, \mh{$W\subseteq V(G)$ s.t. the induced subgraph $\la W\ra$  is connected} \\
							& and, an arbitrary vertex $v_0\in W$. \hfill\\
\emph{Output:} & Relation $\R_{|S_v}(W)^*$. \hfill
\end{tabular}
	\begin{enumerate}
		\item\label{alg:delta0} Initialization.
		\item\label{alg:delta1} Set sequence $Q$ of vertices $v_0,v_1,\dots,v_n$ that form $W$ in BFS-order with respect to $v_0$.
		\item\label{alg:delta2} Set $W':=\emptyset$.
		\item\label{alg:delta3} Assign pairwise different temporary global colors to edges incident to $v_0$.
		\item\label{alg:delta4} \For\ any vertex $v_i$ from sequence $Q$ \Do:
		\begin{enumerate}
			\item\label{alg:delta4a} Use Algorithm \ref{alg:PSPrecognition} to compute $\R_{|S_v}(W'\cup v_i)^*$.
			\item\label{alg:delta4b} Add $v_i$ to $W'$.
		\end{enumerate}
	\end{enumerate}
\end{algo}
\normalsize


In order to show that Algorithm \ref{alg:PSPrecognition} \w{correctly} recognizes
 the local coloring, we define the (temporary) relations $\alpha_c$ and
$\beta_c$ for a \w{chosen} vertex $c$:
Two \emph{primal} edges of $S_c$ are
\begin{itemize}
	\item in \emph{relation $\alpha_c$} if they are
					 contained in the incidence list  and
	\item in \emph{relation $\beta_c$} if they are contained in the absence list
\end{itemize}
after Algorithm \ref{alg:PSPrecognition} is
 \w{executed}  for $c$.
Note, we denote by $ \notalpha_c$
the complement of $\alpha_c$\w{,} which contains all pairs of primal edges of
PSP $S_c$ that are not listed in the incidence list.

\begin{lem}\label{lem:alpha_1new}
Let $e$ and $f$ be two primal edges of the PSP $S_c$.
If $e$ and $f$ span a square with some non-primal vertex $w$ as
unique top-vertex, then $(e,f) \in \alpha_c$.
\end{lem}
\begin{proof}
Let $e=cu_1$ and $f=cu_2$ be primal edges in $S_c$ that span a square
$cu_1wu_2$ with unique top-vertex $w$, where $w$ is non-primal.
Note, since $w$ is the unique top vertex, the vertices $u_1$ and $u_2$ are its only primal neighbors.
W.l.o.g. assume that for vertex $w$ no first primal neighbor was assigned
and let first $u_1$ and then $u_2$ be visited. In Step
\ref{alg:PSPrecognition1a} vertex $w$ is recognized and the first primal
neighbor $u_1$ is determined in Step \ref{alg:PSPrecognition1aii}. Take
the next vertex $u_2$. Since $w$ is not primal and was already visited, we
are in Step \ref{alg:PSPrecognition1aiii}. Since only one primal neighbor
of $w$ was recognized so far, we go to Step \ref{alg:PSPrecognition1aiiiA}.
If $(cu_1,cu_2)$ is not already contained in the incidence list, it
will be added now and thus, $(cu_1,cu_2)\in \alpha_c$.
\end{proof}

\begin{cor}
\label{cor:alpha_2}
Let $e$ and $f$ be two adjacent distinct primal edges of the PSP $S_c$.
If $(e,f)\in \notalpha_c$, then $e$ and $f$ do not span a square or
span a square with non-unique \tk{or primal} top vertex.
In particular, $\notalpha_c$ contains all pairs $(e,f)$
that do not span any square.
\end{cor}
\begin{proof}
The first statement is just the contrapositive of the statement in Lemma \ref{lem:alpha_1new}.
For the second statement observe that
if $e=cx$ and $f=cy$ are two distinct primal edges of $S_c$ that
do not span a square, then the vertices $x$ and $y$ do not have a common
non-primal neighbor $w$. It is now easy to verify that in
none of the substeps of Step \ref{alg:PSPrecognition1}
the pair $(e,f)$ is added to the incidence list,
and thus, $(e,f)\in  \notalpha_c$.
\end{proof}

\begin{lem}\label{lem:beta_1}
Let $e$ and $f$ be two primal edges of the PSP $S_c$ that are in relation
$\beta_c$. Then $e$ and $f$ do not span a unique chordless square with unique
top vertex.
\end{lem}
\begin{proof}
Let $e=cu_1$ and $f=cu_2$ be primal edges of $S_c$. Then pair $(e,f)$
is inserted to absence list in:
\begin{itemize}
\item[a)] Step \ref{alg:PSPrecognition1ai}, when $u_1$ and $u_2$ are
          adjacent. Then  \w{no} square spanned by $e$ and $f$ can be
          chordless.
\item[b)] Step \ref{alg:PSPrecognition1aiiiA} (\Else-condition),
					when $(e,f)$ is already
          listed in the incidence list and another square spanned by $e$ and
          $f$ is recognized. Thus, $e$ and $f$ do not span a unique square.
\item[c)] Step \ref{alg:PSPrecognition1aiiiB}, when $e$ and $f$ span a square
          with top vertex $w$ that has more than two primal neighbors and
          at least one of the primal vertices $u_1$ and $u_2$ are recognized as
          first or second primal neighbor of $w$. Thus $e$ and $f$ span \w{a}
          square with non-unique top vertex.
\end{itemize}
\end{proof}

\begin{lem}\label{lem:beta_2}
Relation $\beta^*_c$ contains all pairs of primal edges $(e,f)$ of $S_c$
that satisfy at least one of the following conditions:
\begin{itemize}
\item[a)] $e$ and $f$ span \w{a} square with \w{a} chord.  
\item[b)] $e$ and $f$ span \w{a} square with non-unique top vertex. 
\item[c)] $e$ and $f$ span more than one square.
\end{itemize}
\end{lem}
\begin{proof}
Let $e=cu_1$ and $f=cu_2$ be primal edges of w{[PSP]} $S_c$.
\begin{itemize}
\item[a)] If $e$ and $f$ span \w{a} square with \w{a} chord, then $u_1$ and $u_2$ are
          adjacent or the top vertex $w$ of the spanned square is primal and thus,
					there is a primal edge $g=cw$. In the first case, we can conclude
					analogously as in the proof of Lemma
          \ref{lem:beta_1} that  $(e,f)\in \beta_c$.
					In \w{the} second case, we \w{analogously} obtain $(e,g), (f,g) \in \beta_c$
					and therefore,           $(e,f)\in\beta_c^*$.
\item[b)] Let $e$ and $f$ span a square with non-unique top vertex $w$.
					If at
          least one of the primal vertices $u_1,u_2$ is a first or second
          neighbor of $w$ then $e$ and $f$ are listed in the absence
          list, as shown in the proof of Lemma \ref{lem:beta_1}. If
          $u_1$ and $u_2$ are neither first nor second primal neighbor\w{s} of
          $w$\w{,} then both edges $e$ and $f$ will be added to the absence list in Step
          \ref{alg:PSPrecognition1aiiiB}\w{,} together with the primal edge $g=cu_3$, where
 					$u_3$ is the first recognized primal
          neighbor of $w$. In other words, $(e,g),(f,g)\in\beta_c$ and hence,
          $(e,f)\in\beta_c^*$.
\item[c)] Let $e$ and $f$ span two squares with top vertices $w$ and $w'$,
          respectively and assume w.l.o.g. that first vertex $w$ is visited and then
					$w'$.
          If both vertices $u_1$ and $u_2$ are recognized as first and
          second primal neighbor\w{s} of $w$ and $w'$, then $(cu_1,cu_2)$ is added to
					the incidence list when visiting $w$ in Step
          \ref{alg:PSPrecognition1aiiiA}. However, when we visit $w'$\w{,} then
 					we insert $(cu_1,cu_2)$ to the absence list in Step
          \ref{alg:PSPrecognition1aiiiA},  \w{because} this pair is
	        already included in the incidence list.
					Thus, $(e,f)\in \beta_c$.
					If at least one of the vertices
          $w,w'$ does not have $u_1$ and $u_2$ as first or second
          primal neighbor\w{,} then $e$ and $f$ must span a square with non-unique
          top vertex. Item b) implies that $(e,f)\in\beta_c^*$.
\end{itemize}
\end{proof}

\begin{lem}\label{lem:beta_3}
Let $f$ be a non-primal edge and $e_1,e_2$ be two distinct primal edges of $S_c$.
Let $(e_1,f),(e_2,f)\in \R_c$. Then $(e_1,e_2)\in\beta^*_c$.
\end{lem}
\begin{proof}
Since the edge $f$ is non-primal, $f$ is not incident with the center $c$.
Recall, by the definition of $\R_c$, two distinct edges can be in relation $\R_c$ only if they have \w{a} common vertex
or are opposite edges in a square. To prove our lemma we need to investigate
\w{the} three following cases, which are also \w{illustrated} in Figure \ref{fig:proof_beta1}:
\begin{itemize}
\item[a)] Suppose both edges $e_1$ and $e_2$  \w{are incident} with $f$.
					Then $e_1$ and $e_2$ span \w{a} triangle and consequently $(e_1,e_2)$ will be
					added to the absence list in Step \ref{alg:PSPrecognition1ai}.
\item[b)] Let $e_1$ and $e_2$ be opposite to $f$ in some squares. There are
          two possible cases (see Figure \ref{fig:proof_beta1} b)). In the
          first case $e_1$ and $e_2$ span a square with non-unique top
          vertex. By Lemma \ref{lem:beta_2}, $(e_1,e_2)\in\beta^*_c$. In the
          second case $e_1$ and $e_2$ span triangles with other primal
          edges $e_3$ and $e_4$. As in  \w{Case} a) of this proof, we have $(e_1,e_3)\in\beta_c$,
          $(e_3,e_4)\in\beta_c$, $(e_4,e_2)\in\beta_c$ and consequently,
          $(e_1,e_2)\in\beta^*_c$
\item[c)] Suppose only  $e_1$ has a common vertex with $f$
					and $e_2$ is opposite to
          $f$ in a square. Again we need to consider two cases (see
          Figure \ref{fig:proof_beta1} c)).
					Since $e_1$ and $f$ are adjacent and $(e_1,f)\in \R_c$,
					we can conclude that either no square is spanned by $e_1$ and $f$\w{,}			
				 or \w{that} the square spanned by $e_1$ and $f$
          is not chordless or not unique.
					Its easy to see that in the first case the edges $e_1$ and $e_2$
          are contained in a common triangle and thus will be added to the absence list
					in Step \ref{alg:PSPrecognition1ai}.
					In the second case $e_1,e_2$ span a
          square which has
          a chord or \tk{has a non-unique top vertex}. In both cases
          Lemma \ref{lem:beta_2} implies that $e_1$ and $e_2$ are in relation
          $\beta^*_c$.
\end{itemize}
\end{proof}

\begin{figure}[tbp]
\centering
\includegraphics[bb= 128 286 472 501, scale=.8]{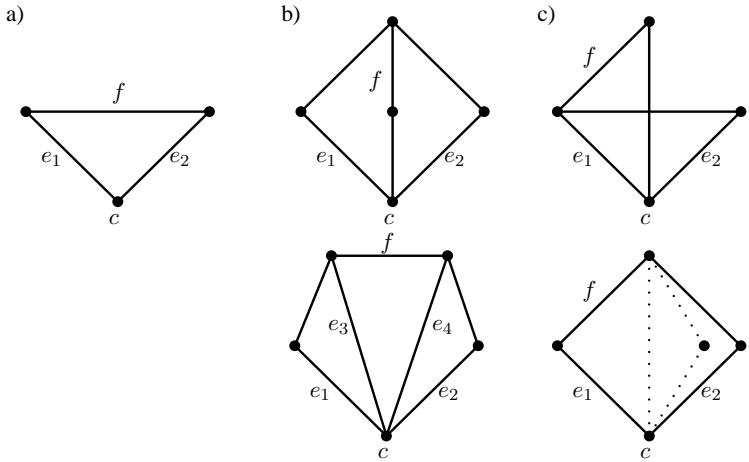}
\caption{The three \w{possible} cases a), b) and c) that are investigated in the
					Proof of Lemma \ref{lem:beta_3}.}\label{fig:proof_beta1}
\end{figure}

\begin{lem}
Let $e$ and $f$ be distinct primal edges of the PSP $S_c$.
Then $(e,f)\in (\notalpha_c\cup\beta_c)^*$ if and only if $(e,f)\in\R_c^*$.
\label{lem:primalIFF}
\end{lem}
\begin{proof}
Assume first that $(e,f)\in \notalpha_c\cup\beta_c$.
By Corollary \ref{cor:alpha_2}, if $(e,f)\in \notalpha_c$\w{,} then $e$ and $f$
do not span a common square\w{,} or span a square with non-unique \tk{or primal} top vertex. In the first case, $e$ and $f$ are in relation $\delta_G$ and consequently also in
relation $\R_c$. On the other hand, if $e$ and $f$ span square with
non-unique top vertex then\w{,} by Lemma \ref{lem:PSP1}, $e$ and $f$ are in
relation $\R_c^*$ as well. \tk{Finally, if $e$ and $f$ span a square with primal
top vertex $w$\w{,} then this square has a chord $cw$ and $(e,f)\in\R_c^*$}.
If $(e,f)\in\beta_c$\w{,} then Lemma \ref{lem:beta_1} implies that $e$ and $f$ do not
span unique chordless square with unique top vertex. Again, by Lemma \ref{lem:PSP1}\w{,}
we  \w{infer that} $(e,f)\in\R_c^*$. Hence, $\notalpha_c\cup\beta_c \subseteq \R_c^*$\w{,}
and consequently, $(\notalpha_c\cup\beta_c)^*\subseteq \R_c^*$.

Now, let $(e,f)\in\R_c^*$. Then there is a sequence
$U=(e=e_1,e_2,\dots,e_k=f)$, $k\geq2$\w{,} with $(e_i,e_{i+1})\in\R_c$ for
$i=1,2,\dots,k-1$. By definition of $\R_c$, two primal edges
are in relation $\R_c$ if and only if they do not span a unique and chordless
square. Corollary \ref{cor:alpha_2} and Lemma \ref{lem:beta_2} imply that
all these pairs are contained in $(\notalpha_c\cup\beta_c)^*$. Hence, any
two consecutive primal edges $e_i$ and $e_{i+1}$ contained in the sequence
$U$ are in relation $(\notalpha_c\cup\beta_c)^*$. Assume \w{that} there is an edge
$e_i \in U$ \w{that} is not incident to \w{the} center $c$ and thus,
non-primal. By \w{the} definition of $\R_c$\w{,} and since $(e_{i-1},e_i),(e_i,e_{i+1})\in\R_c$, we can
conclude that the edges $e_{i-1}$ and $e_{i+1}$ must be primal in $S_c$.
Lemma \ref{lem:beta_3} implies that $e_{i-1}$ and $e_{i+1}$ must be in
relation $\beta^*_c$. By removing all edges from $U$ that are not incident
with $c$ we \w{obtain} a sequence $U'=e=e_1,e'_2,\dots,e'_j=f$ of primal edges. By
analogous arguments as before, all pairs $(e'_i,e'_{i+1})$ of $U'$ must be contained
in $(\notalpha_c\cup\beta_c)^*$. By transitivity,
$e$ and $f$ are also in $( \notalpha_c\cup\beta_c)^*$.
\end{proof}

\begin{cor}\label{cor:localColoringDet}
Let $e$ and $f$ be primal edges of \w{the} PSP $S_c$. Then $(e,f)\in (
\notalpha_c\cup\beta_c)^*$ if and only if $e$ and $f$ have the same local
color in $S_c$.
\end{cor}
\begin{proof}
	This is an immediate consequence of Lemma \ref{lem:primalIFF},
	the local color assignment\w{,} and the merging procedure
  (Step \ref{alg:PSPrecognition2} and \ref{alg:PSPrecognition3}) in
	Algorithm  \ref{alg:PSPrecognition}.
\end{proof}


\begin{lem}\label{lem:PSP_recognition}
Let $\R_{|S_v}(W)^*$ be a global coloring associated with a set of treated
centers $W$ and assume that the induced subgraph $\la W \ra$ is connected.
Let $c$ be a vertex that is not contained in $W$ but \w{adjacent to a vertex in $W$}.
Then Algorithm \ref{alg:PSPrecognition} computes the global coloring
$\R_{|S_v}(W\cup c)^*$ by taking $W$ and $c$ as input.
\end{lem}
\begin{proof}
Let $W\subseteq V(G)$ be a set of PSP centers and let $c\in V(G)$ be a
given center of PSP $S_c$ where $c\not\in W$ \tk{and $\langle W\cup c\rangle$ is connected}.
In Step \ref{alg:PSPrecognition1} of Algorithm \ref{alg:PSPrecognition}
we compute the absence and incidence lists. In Step \ref{alg:PSPrecognition2},
we assign pairwise different temporary local colors to any primal edge adjacent
to $c$. Two
temporary local colors $b_1$ and $b_2$ are then merged in Step
\ref{alg:PSPrecognition3} if and only if there exists some pair of primal
edges $(e_1,e_2)\in(\notalpha_c\cup\beta_c)$ where
$e_1$ is colored with $b_1$ and $e_2$ with $b_2$. \w{Therefore,}
\w{merged}
temporary local colors reflect equivalence classes
of $(\notalpha_c\cup\beta_c)^*$ containing the primal edges incident to $c$.
By Corollary \ref{cor:localColoringDet}, $(\notalpha_c\cup\beta_c)^*$
classes indeed determine the local colors of primal edges
in $S_c$.

Note, if one knows the colors of primal edges incident
to $c$, \w{then} it is very easy to determine the set of
non-primal edges of $S_c$, as any two primal edges of different
equivalence classes span a unique and chordless square.
In Step \ref{alg:PSPRecognition5}, we investigate each vertex $v$ from
stack $\mathbb{T}_c$ and check the local colors of primal edges $cw_1$ and $cw_2$\w{,}
where $w_1$ and $w_2$ are the first and second recognized primal neighbors of
$v$, respectively. If $cw_1$ and $cw_2$ differ in \w{their} local colors\w{,}
then $vw_1$ and $vw_2$ are non-primal edges of $S_c$\w{, as}  follows from the PSP
construction. Recall that the stack contains all vertices that are at distance
two from center $c$ and which are adjacent to at least two primal vertices. In
other words, the stack contains all \tk{non-primal} top vertices of all squares spanned by
primal edges. Consequently, we claim that all non-primal edges of the PSP $S_c$
are treated in Step \ref{alg:PSPRecognition5}. Note that non-primal edges have
the same local color as \w{their}  opposite primal edge, which is unique by Lemma
\ref{lem:PSP2}.

As we already argued, after Step \ref{alg:PSPrecognition3} is performed we
know\w{,} or can \w{at least} easily determine all edges of $S_c$ and their local colors.
Recall that local colors define the local coloring $\R_{|S_c}$. Suppose,
temporary global colors that correspond to the global coloring
$\R_{|S_v}(W)^*$ are assigned.
Our goal is to modify and identify temporary global colors such that
they will correspond to the global coloring $\R_{|S_v}(W\cup c)^*$. Let
$B_1,B_2,\dots,B_k$ be \w{the} classes of $\R_{|S_c}$ (\emph{local classes}) and
$D_1,D_2,\dots,D_l$ be \w{the} classes of $\R_{|S_v}(W)^*$
(\emph{global classes}).
When a local class $B_i$ and a
global class $D_j$ have a nonempty
intersection\w{, then} we can infer that all their edges must be contained in a
common class of $\R_{|S_v}(W\cup c)^*$.
Note, by means of Lemma \ref{lem:vMeetsEveryClass},
we can conclude that for each local class $B_i$ there is a global class $D_j$ such
that $B_i\cap D_j\neq \emptyset$, see also \cite{HeImKu-2013}.
In that case we need to guarantee that edges of
$B_i$ and $D_j$ will be colored by the same temporary global color. Note,
in the beginning of the iteration two edges have the same temporary global
color if and only if they lie in a common global class.

In Step \ref{alg:PSPRecognition4} and Step \ref{alg:PSPRecognition5}, we
investigate all primal and non-primal edges of $S_c$. When we treat first
edge $e$ that is colored by some local color $b_i$, that is $e\in B_i$,
and has already \w{been} assigned some temporary global color
$d_j$, and \w{therefore}  $e\in D_j$,  then we map $b_i$
to $d_j$. \w{Thus},  we keep the information that $e\in B_i\cap D_j$.
In Step \ref{alg:PSPRecognition6}, we then assign temporary global color $d_j$
to any edge of $S_c$ that is colored by the local color $b_i$. If
 the local color $b_i$ is already mapped to some temporary global color
$d_j$\w{,} and \w{if} we find another edge of $S_c$ that is colored by $b_i$ and
simultaneously has \w{been} assigned some different temporary global color $d_{j\prime}$\w{,}
then we merge $d_j$ and $d_{j\prime}$ in Step \ref{alg:PSPRecognition4i}.
Obviously this is correct, since $B_i\cap D_j\neq\emptyset$ and $B_i\cap
D_{j\prime}\neq\emptyset$\w{,} and hence $D_j,D_{j\prime}$ and $B_i$ must be
contained in a common  equivalence class of $\R_{|S_v}(W\cup c)^*$.
\w{Recall,}
for each local class $B_i$ there is a global class $D_j$ such that $B_i\cap
D_j\neq \emptyset$. \w{This} means that every local color is mapped to some
global color\w{,} and \w{consequently there is no}
need to create \w{a} new temporary global color
in Step \ref{alg:PSPRecognition6}.

Therefore, whenever local and global classes  \w{share an} edge\w{, then}
all their edges will have the same temporary global color at the end of Step
\ref{alg:PSPRecognition6}. On the other hand\w{,} when edges of two different
global classes are colored by the same temporary global color\w{,} then both
global classes must be contained in a common class of
$\R_{|S_v}(W\cup c)^*$.

Hence, after \w{the performance of} Step \ref{alg:PSPRecognition6},
the merged temporary global colors determine \w{the equivalence} classes of $\R_{|S_v}(W\cup c)^*$.
\end{proof}


\begin{lem}
Let $G$ be a connected graph, \mh{ $W\subseteq V(G)$ s.t. $\la W \ra$ is connected}, and 
$v_0$ an arbitrary vertex of $G$. Then
Algorithm \ref{alg:deltaRecognition} computes the global coloring
\mh{$\R_{|S_v}(W)^*$} by taking $G$, $W$ and $v_0$ as input.
\end{lem}
\begin{proof}
In Step \ref{alg:delta1} we define the BFS-order in which
\w{the vertices}
will be processed and store this sequence in $Q$. In Step \ref{alg:delta3} we assign pairwise different
temporary global colors to all edges that are incident with $v_0$. In Step
\ref{alg:delta4} we iterate over all vertices of \mh{the given induced connected 
subgraph $\la W\ra$ of $G$}. For
every vertex we execute Algorithm \ref{alg:PSPrecognition}. Lemma
\ref{lem:PSP_recognition} implies that in the first
iteration, we correctly compute the local colors for $S_{v_0}$\w{,} and consequently also
\mh{$\R_{|S_v}(\{v_0\})^*$}. Obviously, whenever we merge two
temporary local colors of two primal edges in the first iteration\w{, then} we also merge
their temporary global colors.
Consequently, the resulting temporary global colors correspond to
the global coloring \mh{$\R_{|S_v}(\{v_0\})^*$} after the first iteration. Lemma
\ref{lem:PSP_recognition} implies that after all iterations are performed, \mh{that is,  
all vertices in $Q$ are processed,}
the resulting temporary global colors correspond to \mh{$\R_{|S_v}(W)^*$ for the given
input set $W\subseteq V(G)$}.
\end{proof}

\w{For the global coloring, Theorem \ref{thm:union_equals_delta} implies that
 $\R_{|S_v}(V(G))^* = \delta_G^*$. This leads to the following corollary.}

\begin{cor}
Let $G$ be a  \w{connected} graph and $v_0$  an arbitrary vertex of $G$. Then
Algorithm \ref{alg:deltaRecognition} computes the global coloring
$\delta_G^*$ by taking $G$, \mh{$V(G)$} and $v_0$ as input.
\label{cor:deltastar}
\end{cor}

\subsection{Time Complexity}

We begin with the complexity of merging colors. We have global and local
colors, and will define \emph{local} and \emph{global color graphs}. Both
graphs are acyclic temporary structures. Their vertex sets are the sets of
temporary colors in the initial state. In this state the color graphs have
no edges. Every component is a single vertex and corresponds to an initial
temporary color. Recall that we color edges of graphs, for example the
edges of $G$ or $S_v$. The color of an edge is indicated by a pointer to a
vertex of the color graph. These pointers are not changed, but the colors
will correspond to the components of the color graph. When two colors are
merged, then this reflected by adding an edge between their respective
components.

The color graph is represented by an adjacency list as described in
\cite[Chapter 17.2]{haimkl-2011} or \cite[pp. 34 -37]{imkl-2000}.
Thus, working with the color graph needs $O(k)$ space
when $k$ colors are used.
Furthermore, for every vertex of the color graph we keep an index of the
connected component in which the vertex is contained. We also store the
actual size of every component, that is, the number of vertices of this
component.

Suppose we wish to merge temporary colors of edges $e$ and $f$ that
are identified with vertices $a$, respectively $b$, in the color graph.
We first check whether $a$ and $b$ are contained in the same connected
component by comparing component indices. If the component indices are the
same, then $e$ and $f$ already have the same color, and no action is
necessary. Otherwise we insert an edge between $a$ and $b$ in the color
graph. As this merges the components of $a$ and $b$ we have to update
component indices and the size. The size is updated in constant time. For
the component index we use the index of the larger component. Thus, no
index change is necessary for the larger component, but we have to assign
the new index to all vertices of the smaller component.

Notice that the color graph remains acyclic, as we only add edges between different components.

\begin{lem}
Let $G_0 =(V,E)$ be a graph with $V=\{v_1,\dots,v_k\}$ and $E=\emptyset$.
The components of $G_0$ consist of single vertices. We assign component index
$j$ to every component $\{v_j\}$.
For $i=1, \dots , k-1$ let $G_{i+1}$ denote the graph that
results from $G_{i}$ by adding an edge between two distinct
connected components, say $C$ and $C'$. If $|C| \leq |C'|$, we use the
the component index of $C'$ for the new component and assign it
to every vertex of $C$.

Then every $G_i$ is acyclic, and the total cost of merging colors is $O(k \log_2 k)$.
\label{lem:log2}
\end{lem}

\begin{proof}
Acyclicity is true by construction.

A vertex is assigned a new component index when its component is merged
with a larger one. Thus, the size of the component at least doubles at
every such step. Because the maximum size of a component is bounded by $k$,
there can be at most $\log_2 k$ reassignments of the component index for
every vertex. As there are $k$ vertices, this means that the total cost of
merging colors is $O(k \log_2 k)$.
\end{proof}

The color graph is used to identify \tk{temporary} local, resp., global colors. Based on this,
we now define the \emph{local} and \emph{global color graph}.

Assigned labels of the vertices of the global color graph are stored in the edge
list, where any edge is identified with at most one such label.
A graph is represented
by an extended adjacency list, where for any vertex and its neighbor
a reference to the edge (in the edge list) that connects them is stored. This
reference allows to access a global temporary color from adjacency list in
constant time. 

In every iteration of Algorithm \ref{alg:deltaRecognition}, we recognize the PSP
for one vertex by calling Algorithm \ref{alg:PSPrecognition}.
\w{In the following paragraph we introduce several temporary
attributes and matrices that are used in the algorithm.}

 Suppose we execute  \w{an} iteration that recognizes some PSP $S_c$.
To indicate whether a vertex was treated in this iteration we  \w{introduce the} attribute $visited$,
that is, when vertex $v$ is visited in this iteration we set $v.visited=c$.
Any value different from $c$ means
that vertex $v$ was not yet treated in this iteration. Analogously, we
 \w{introduce} the attribute
$primal$ to indicate that a vertex is adjacent to the current center $c$. The attribute
$tempLabel$ maps primal vertices to the indices of rows and columns of the
matrices $incidenceList$ and $absenceList$. For any vertex $v$ that is at
distance two from the center $c$ we store its first and second primal neighbor
$w_1$ and $w_2$ in the attributes $FirstPrimalNeighbor$ and
$SecondPrimalNeighbor$. Furthermore, we need to keep the position of $vw_1$ and
$vw_2$ in the edge list to get their temporary global colors. For this
purpose, we use attributes $firstEdge$ and $secondEdge$. \tk{Attribute}
$mapLocalColor$ helps us to map temporary local colors to the vertices of
the global color graph. Any vertex that is at distance two from the center and
has a least two primal neighbors is a candidate for a non-primal vertex. We
insert them to the $stack$. The temporary structures help to access the required
information in constant time:


\begin{itemize}
	\item $v.visited=c$ \\ vertex $v$ has been already visited in the current iteration.
	\item $v.primal=c$ \\ vertex $v$ is adjacent to center $c$.
	\item $incidenceList[v.tempLabel,u.tempLabel]=0$ \\ pair of primal edges $(cv,cu)$ is missing in the incidence list.
	\item $absenceList[v.tempLabel,u.tempLabel]=1$ \\ pair of primal edges $(cv,cu)$ was inserted to the absence list.
	\item $v.firstPrimalNeighbor=u$ \\ $u$ is the first recognized primal neighbor of the non-primal vertex $v$.
	\item $v.firstEdge=e$ \\ edge $e$ joins the non-primal vertex $v$ with its first recognized primal neighbor
					(it is used to get \w{the} temporary global color from \w{the} edge list).
	\item $b.mapLocalColor=d$ \\ local color $b$ is mapped to temporary global color $d$ (i.e. there exists an edge that is colored by both colors).
\end{itemize}

Note that the temporary matrices $incidenceList$ and $absenceList$ have
dimension $\deg(c)\times \deg(c)$  and \w{that} all their entries are set to zero in the
beginning of every iteration.
\begin{thm}
For a given connected graph $G=(V,E)$ with maximum degree $\Delta$ and $W\subseteq V $,
Algorithm \ref{alg:deltaRecognition} runs in $O(|E|\Delta)$ time and $O(|E|+\Delta^2)$ space.
\label{thm:compl}
\end{thm}
\begin{proof}
Let $G$ be a given graph with $m$ edges and $n$ vertices.
In Step \ref{alg:delta0} of Algorithm
\ref{alg:deltaRecognition} we initialize all temporary attributes and
matrices. This consumes $O(m+n) = O(m)$ time and space,
since $G$ is connected\w{,} and hence, $m\geq n-1$. Moreover\w{,} we set all
temporary colors of edges in the edge list to zero,  which does not increase the
time and space complexity of the initial step. Recall that we use an extended adjacency
list, where every vertex and its neighbors keep the reference to the edge in
the edge list that connects them. To create an extended adjacency list we iterate over
all edges in the edge list\w{,} and for every edge $uv=e\in E(G)$ we set a new
\w{entry for the} neighbor $v$ for $u$ and\w{,} simultaneously\w{,} we add a reference $v.edge=e$.
The same is done for vertex $v$. It can be done in $O(m)$ time and space.

In Step \ref{alg:delta1} of Algorithm \ref{alg:deltaRecognition},
we build a sequence of vertices in BFS-order
starting with $v_0$, which is done in $O(m+n)$ time in general.
Since $G$ is connected, the BFS-ordering can be computed in $O(m)$ time.
Step \ref{alg:delta2} takes constant time.
In Step \ref{alg:delta3} we
initialize the global color graph that has $\deg(v_0)$ vertices (\w{bounded by} $\Delta$ in
general). As we already showed, all operations on the global color graph take
$O(\Delta\log_2{\Delta})$ time and $O(\Delta)$ space. We proceed to traverse
all neighbors $u_1,u_2,\dots,u_{\deg{(v_0)}}$ of the root $v_0\in V(G)$
(via the adjacency list) and  \w{assign} them unique labels $1,2,\dots,\deg(v_0)$ in edge list,
that is,  every edge $v_0u_i$ gets the label $i$.
In this way, we initialize pairwise different temporary global colors of edges
incident with $v_0$ \w{, that is, to vertices} of the global color graph.
\w{Using} the extended adjacency list\w{,} we set the label to an edge in
the edge list in constant time. In Step \ref{alg:delta4} we run Algorithm
\ref{alg:PSPrecognition} for any vertex from the defined BFS-sequence.

In the remainder of this proof, we will focus on \w{the} complexity of
Algorithm \ref{alg:PSPrecognition}.
Suppose we perform Algorithm \ref{alg:PSPrecognition} for
vertex $c$ to recognize the PSP $S_c$. The recognition process is based on
temporary structures. We do not need to reset any of these structures, for any \w{execution}
of Algorithm \ref{alg:PSPrecognition} \w{for a new center $c$}, except $absenceList$ and
$incidenceList$. \tk{This is done in Step \ref{alg:PSPrecognition0}. 
Further, we set here} the attribute $tempLabel$ for
every primal vertex $v$, such that every vertex has assigned a unique number
from $\{1,2,\dots,\deg(c)\}$. \tk{Finally}, we traverse
all neighbors of the center $c$ and for each of them we set $primal$ to $c$.
Hence, the initial step of Algorithm \ref{alg:PSPrecognition} is done in
$O(\deg(c)^2)$ time.

Step \ref{alg:PSPrecognition1a} is performed for every neighbor of every
primal vertex. The number of all such neighbors is at most $\deg(c)\Delta$.
For every treated vertex, we set attribute $visited$ to $c$. This allows
us to verify in constant time that a vertex was already visited in the
recognition subroutine Algorithm \ref{alg:PSPrecognition}.

If the condition in Step \ref{alg:PSPrecognition1ai} is satisfied, then we
put primal edges $cu$ and $cw$ to the absence list. By the previous arguments,
this can  be done in constant time by usage of $tempLabel$ and $absenceList$.

If the condition in Step \ref{alg:PSPrecognition1aii} is satisfied, we set
vertex $u$ as first primal neighbor of vertex $w$. For this purpose, we
use the attribute $firstPrimalNeighbor$. We also set $w.firstEdge=e$\w{,} where $e$
is a reference to the edge in the edge list that connects $u$ and $w$. This
reference is obtained from the extended adjacency list in constant time. Recall,
the edge list is used to store the labels of vertices of
the global color graph for the edges of a given graph, that is, the assignment
of temporary global colors to the edges.
Using $w.firstEdge$, we are able to directly
access the temporary global color of edge $uw$ in constant time.

Step \ref{alg:PSPrecognition1aiii} is performed when we try to visit a vertex
$w$ from some vertex $u$ where $w$ has been already visited before from some
vertex $v$. If $v$ is the only recognized primal neighbor of $w$\w{, then we} we perform
analogous operations as in the previous step. Moreover. if $(cu,cv)$ is not contained
in the incidence list, then we set $u$ as second primal neighbor of $w$, \w{add}
$(cu,cv)$ to the incidence list and \w{add}  $w$ to the stack. Otherwise we add
$(cu,cv)$ to the absence list. The number of operations in this step is
constant.

If $w$ has more recognized primal neighbors we process case B. Here we just add all pairs formed
 by $cv_1,cv_2,cu$ to absence list. Again\w{,} the number of operations is constant by usage of $tempLabel$
 and matrices $incidenceList$ and $absenceList$.

In Step \ref{alg:PSPrecognition2}\w{[,]} we assign pairwise different
temporary local colors
to the primal edges. Assume the neighbors of the center $c$ are labeled
by $1,2,\dots,\deg{(c)}$, then we set value $u.tempLabel$ to \tk{$c$}$u$.
In Step \ref{alg:PSPrecognition3a}
we iterate over all entries of the $absenceList$. For all pairs of edges that
are in the absence list we check whether they \w{still} have  different temporary
local colors and if so, we merge their temporary local colors  by adding
a respective edge in the local
color graph. Analogously we treat all pairs of edges contained in the
$incidenceList$ in Step \ref{alg:PSPrecognition3b}.
Here we merge temporary local colors of primal
edges $cu$ and $cv$ when the pair $(cu,cv)$ is missing. To treat all entries of
the $absenceList$ and $incidenceList$ we need to perform $\deg(c)^2$ iterations.
Recall, the temporary local color of the primal edge $cu$ is equal to the
index of the connected component in the local color graph,
in which vertex $u.tempLabel$ is contained.
Thus, the temporary local color of this primal edge can be accessed in
constant time. As we already showed, the number of all operations on the local
color graph is bounded by $O(\deg(c)\log_2{\deg(c)})$.
Hence, the overall time complexity
of both Steps \ref{alg:PSPrecognition2} and \ref{alg:PSPrecognition3} is $O(\deg(c)^2)$.

In Step \ref{alg:PSPRecognition4} we map temporary local colors of primal
edges to temporary global colors. For this purpose, we use the attribute
$mapLocalColor$. The temporary global color of every edge can be accessed by
the extended adjacency list, the edge list and the global color graph in constant time.
Since we need to iterate over all primal vertices, we can conclude that Step
\ref{alg:PSPRecognition4} takes $O(\deg(c))$ time.

In Step \ref{alg:PSPRecognition5} we perform analogous operations for any vertex
from Stack $\mathbb{T}_c$ as in Step \ref{alg:PSPRecognition4}. In the worst
case,  we add all vertices that are at distance two from the center to the
stack.
Hence, the size of the stack is bounded by $O(\deg(c)\Delta)$. Recall that the
first and second primal neighbor $w_1$ and $w_2$ of every vertex $v$ from the stack can
be directly accessed by the attributes $firstPrimalNeighbor$ and
$secondPrimalNeighbor$. On the other hand, the temporary global colors of
non-primal edges $vw_1$ and $vw_2$ can be accessed directly by the attributes
$firstEdge$ and $secondEdge$. Thus, all needful information can be accessed
in constant time. Consequently, the time complexity of this step is bounded by
$O(\deg(c)\Delta)$.	

In the last \w{step,} Step \ref{alg:PSPRecognition6}, we iterate over all edges of the recognized PSP.
Note, the list of all primal edges can be obtained from the extended adjacency list. To get all
non-primal edges we iterate over all vertices from the stack and use the attributes
$firstEdge$ and $secondEdge$, which takes $O(\deg(c)\Delta)$ time.
The remaining operations can be done in constant time.

To summarize, Algorithm \ref{alg:PSPrecognition} runs in
$O(\deg(c)\Delta)$ time.
Consequently, Step \ref{alg:delta4} of Algorithm
\ref{alg:deltaRecognition} runs in
 $O(\sum_{c\in W}\deg(c) \Delta)= O(m\Delta)$ time,
which defines also the total time complexity of Algorithm \ref{alg:deltaRecognition}.
The most space
consuming structures are the edge list and the extended adjacency list ($O(m)$ space) and the
temporary matrices $absenceList$ and $incidenceList$ $(O(\Delta^2)$ space). Hence,
the overall space complexity is $O(m+\Delta^2)$.
\end{proof}

Since quasi Cartesian products are defined as graphs with non-trivial $\delta^*$,
\w{ Theorem \ref{thm:compl}} and Corollary \ref{cor:deltastar} imply the following \w{corollary}.

\begin{cor}
For a given connected graph $G=(V,E)$ with bounded maximum degree
Algorithm \ref{alg:deltaRecognition} (with slight modifications) determines
whether $G$ is a quasi Cartesian product in $O(|E|)$ time and $O(|E|)$ space.
\end{cor}

\subsection{Parallel Processing}
The local approach allows the parallel computation of $\delta^*(G)$ on
multiple processors. Consider a graph $G$ with vertex set $V(G)$. Suppose
we are given a decomposition of $V(G)=W_1\cup W_2\cup\dots\cup W_k$ into
$k$ parts such, that $|W_1|\approx|W_2|\approx\dots\approx|W_k|$, where
the subgraphs induced by $W_1,W_2,\dots,W_k$ are connected\w{,} and the number
of edges whose endpoints lie in different partitions is small (we call such
decomposition {\it good}). Then algorithm \ref{alg:deltaRecognition} can be
used to compute the colorings
$\R_{|S_v}(W_1)^*,\R_{|S_v}(W_2)^*,\dots,\R_{|S_v}(W_k)^*$, where every
instance of the algorithm can run in parallel. The resulting global
colorings are used to compute
$\R_{|S_v}(V(G))^*=(\R_{|S_v}(W_1)^*\cup\R_{|S_v}(W_2)^*\cup\dots\cup\R_{|S_v}(W_k)^*)^*$.
Let us sketch the parallelization.

\bigskip
\begin{algo}[Parallel recognition of $\delta^*$]
\small
\label{alg:parallelDeltaRecognition}
\begin{tabular}{ll}
\emph{Input:} & A graph $G$, and a good decomposition $V(G)=W_1\cup W_2\cup\dots\cup W_k$.\\
\emph{Output:} & Relation $\delta_G^*$.
\end{tabular}
	\begin{enumerate}
		\item\label{alg:parallel1} \w{F}or every partition $W_i$ \w{concurrently} compute global coloring \tk{$\R_{|S_v}(W_i)$}   ($i\in\{1,2,\dots,k\}$):
		\begin{enumerate}
			\item\label{alg:parallel2} Take all vertices of $W_i$ and order them in BFS to get sequence $Q_i$.
			\item\label{alg:parallel3} Set $W':=\emptyset$.
			\item\label{alg:parallel4} Assign pairwise different temporary global colors to edges incident to first vertex in $Q_i$.
			\item\label{alg:parallel5} For any vertex $v$ from sequence $Q_i$ do:
			\begin{enumerate}
				\item\label{alg:parallel5a} Use Algorithm \ref{alg:PSPrecognition} to compute \tk{$\R_{|S_v}(W'\cup v)^*$}.
				\item\label{alg:parallel5c} Put all edges that were treated in previous \w{step} and have at least one endpoint \w{not in} partition $W_i$ to stack $\mathbb{T}_i$.
				\item\label{alg:parallel5b} Add $v$ to $W'$.
			\end{enumerate}
		\end{enumerate}
		\item Run concurrently for every partition $W_i$ to merge all global colorings ($i\in\{1,2,\dots,k\}$):
		\begin{enumerate}
			\item For each edge from stack $\mathbb{T}_i$\w{,} take all its assigned global colors and merge them.
		\end{enumerate}
	\end{enumerate}
\end{algo}
\normalsize

\begin{figure}
\centering
\includegraphics[bb=93 254 515 545, scale=0.7]{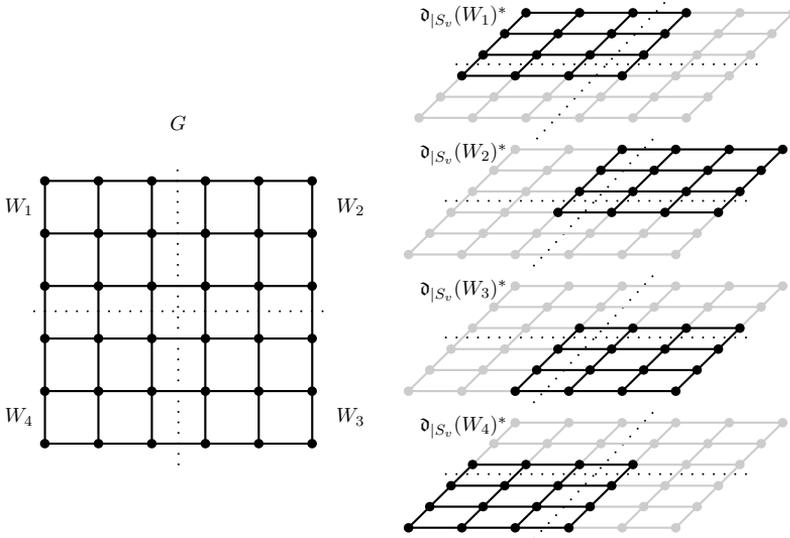}
\caption{Example - Parallel recognition of $\delta^*$.}\label{fig:parallel_delta}
\end{figure}

 Figure \ref{fig:parallel_delta} shows an example of decomposed vertex set
of a given graph $G$. The computation of global colorings associated with
\w{the individual sets of the partition} can be
done then in parallel. The edges that are colored by global color when
\w{the} partition is treated are highlighted by \w{bold} black color. Thus
we can observe that many edges will be colored by more then one color.

\w{Notice that we do not treat} the task of finding a good partition. With the methods of \cite{hazwi-2001}
this is possible with high probability in $O(\log n)$ time, where $n$ is the number of vertices.

\bibliographystyle{plain}
\bibliography{libParRec}

\end{document}